\documentclass{article}

\usepackage{amsmath,amsthm,amssymb,fullpage,color}

\newcounter{constnum}
\renewcommand{\qed}{\nobreak \ifvmode \relax \else
      \ifdim\lastskip<1.5em \hskip-\lastskip
      \hskip1.5em plus0em minus0.5em \fi \nobreak
      \vrule height0.75em width0.5em depth0.25em\fi}

\newtheorem{theorem}{Theorem}
\newtheorem{lemma}{Lemma}
\newtheorem{definition}{Definition}
\newtheorem{corollary}{Corollary}

\begin{document}

\title{On the Non-existence of Perfect Sequences with the Array Orthogonality Property}
\author{Sam Blake\\ {\small\texttt{samuel [dot] thomas [dot] blake [at] gmail [dot] com} }}
\date{\today}
\maketitle

\begin{abstract}
For over three decades, the pursuit of perfect periodic autocorrelation sequences has been constrained by Mow's conjecture, which posits that no perfect sequence over an $n$-phase alphabet can exist with a length greater than $n^2$. While a proof across all conceivable sequence classes remains an open problem, this paper establishes bounds for a prominent class of constructions relying on the Array Orthogonality Property (AOP). We show that sequences generated by pure bivariate polynomial index functions cannot exceed the $n^2$ Frank-Heimiller bound due to algebraic periodicity. Furthermore, we extend this result to floored rational index functions, proving that attempts to geometrically expand the array dimensions inherently result in destructive fractional phase scattering. Neutralising this scattering strictly forces a collapse of the phase space, re-establishing the $n^2$ limit. Finally, we define the boundaries of these theorems, noting their fundamental reliance on commutative algebras, and contrast them with recent sequence constructions demonstrating the existence of unbounded perfect sequences over non-commutative unit quaternions.
\end{abstract}

%
%

\section{Introduction}

The discovery of perfect periodic autocorrelation sequences over complex roots of unity has been a central problem in signal processing, cryptography, and combinatorial design theory. Constructions by Frank-Heimiller \cite{frank1962, heimiller1961}, Milewski \cite{milewski1983}, and Blake \cite{blake2014} yielded perfect sequences by exploiting algebraic phase rotations. A sequence producing a longer orthogonal capacity for a given number of phases would be highly desirable for applications; however, it has long been conjectured by Mow \cite{mow1993} that such sequences do not exist for lengths exceeding $n^2$ for an $n$-phase alphabet. \\

Alternative approaches to proving non-existence have focused heavily on number-theoretic constraints rather than structural geometry. For instance, Ma and Ng \cite{ma2009} demonstrated that the existence of complex $p$-ary perfect sequences (where $p$ is an odd prime) is equivalent to the existence of certain relative difference sets. Using character theory, they proved the absolute non-existence of $p$-ary perfect sequences for several specific lengths---such as $p^s$ for $s \ge 3$, $2p^s$ for $s \ge 1$, and $pq$ for prime $q > p$---regardless of the construction method. \\

While a completely general proof of Mow's conjecture remains elusive, we can establish strict, geometrically driven upper bounds for a broad and natural class of structured sequence constructions. Specifically, many known perfect sequences rely on the Array Orthogonality Property (AOP) applied to an array whose entries are defined by a bivariate polynomial power of a primitive root of unity. This paper shows that sequences generated via the AOP using bivariate polynomial index functions are constrained by the $n^2$ bound. 

%
%

\section{Prerequisites}

Before presenting our main results on the length bounds of sequences that possess the AOP, we explicitly define the AOP and state several foundational conditions for the perfection of sequences and arrays that will serve as the basis of our work.

\begin{definition}[Array Orthogonality Property]
Let $\textbf{s} =$ be a sequence of length $L = RC$.
The array $\textbf{S}$ of size $R \times C$ associated with $\textbf{s}$ is formed by enumerating $\textbf{s}$ row-by-row, such that $S_{i,j} = s_{iC+j}$.
We denote the $j^{\text{th}}$ column of $\textbf{S}$ by $\textbf{S}[j]$. The sequence $\textbf{s}$ possesses the Array Orthogonality Property (AOP) for the divisor $C$ if $\textbf{S}$ satisfies two distinct geometric conditions:
\begin{enumerate}
    \item For all $\tau$ and $j_0 \not\equiv j_1 \mod C$: $\theta_{\textbf{S}[j_0], \textbf{S}[j_1]}(\tau) = 0$.
(That is, any two distinct columns of $\textbf{S}$ are mutually orthogonal.)
\item For all $\tau \not\equiv 0 \mod R$: $\sum_{j=0}^{C-1}\theta_{\textbf{S}[j]}(\tau) = 0$.
(That is, the columns of $\textbf{S}$ form a set of periodic complementary sequences.)
\end{enumerate}
\end{definition}

\begin{theorem}[Mow \cite{mow1993}]\label{thm:aop_perfect_seq}
Any sequence with the AOP is perfect.
\end{theorem}
\begin{proof}
The periodic autocorrelation of a sequence $\textbf{s}$ of length $RC$ for shift $\tau$ is given by 
$$\theta_{\textbf{s}}(\tau) = \sum_{i=0}^{R C - 1} s_i \, s_{i+\tau}^*.$$
Change coordinates by letting $i = q C + r$ (where $0 \le r < C$), and $\tau = q' C + r'$ (where $0 \le r'<C$).
Then we have 
\begin{align*}
\theta_{\textbf{s}}(q'C+r') &= \sum_{r=0}^{C - 1} \sum_{q=0}^{R - 1} s_{q C+r} \, s_{(q+q') C+r+r'}^*\\
&=  \sum_{r=0}^{C - 1} \sum_{q=0}^{R - 1} s_{q C+r} \, s_{\left(q+q'+\left\lfloor \frac{r+r'}{C}\right\rfloor\right) C + (r+r' \bmod C)}^*\\
&= \sum_{r=0}^{C - 1} \sum_{q=0}^{R - 1}  S_{q,r} \, S_{q+q'+\left\lfloor \frac{r+r'}{C}\right\rfloor, (r+r' \bmod C)}^*\\
&= \sum_{r=0}^{C - 1} \theta_{\textbf{S}[r],\textbf{S}[r+r' \bmod C]} \left(q' + \left\lfloor \frac{r+r'}{C}\right\rfloor \right).
\end{align*}
For $r'\neq 0$, the indices $r$ and $r+r' \bmod C$ represent distinct columns. Condition 1 of the AOP dictates that these distinct columns are orthogonal, which implies $\theta_{\textbf{s}}(\tau) = 0$. Otherwise, for $r'=0$, the shift simplifies to $q'C$, and the summation evaluates the sum of the autocorrelations of the individual columns at shift $q'$. Because we are evaluating the off-peak autocorrelation of the sequence, $\tau \not\equiv 0 \mod{RC}$, which implies $q' \not\equiv 0 \mod R$. Thus, condition 2 of the AOP implies $\theta_{\textbf{s}}(\tau) = 0$. Therefore, $\textbf{s}$ has perfect periodic autocorrelation.
\end{proof}

\begin{corollary}\label{thm:aop_array_perfect}
If a sequence $\textbf{s}$ has the AOP, then the array $\textbf{S}$ associated with $\textbf{s}$ has perfect periodic autocorrelation.
\end{corollary}
\begin{proof}
By the first condition of the AOP, any two distinct columns of $\textbf{S}$ are orthogonal, which implies the 2-dimensional off-peak autocorrelation $\theta_{\textbf{S}}(h,v) = 0$ for any horizontal shift $h \neq 0$. Otherwise, if $h = 0$, we consider an off-peak vertical shift $v \not\equiv 0 \mod R$. By the second condition of the AOP, we have
$$\sum_{j=0}^{C-1}\theta_{\textbf{S}[j]}(v) = \sum_{j=0}^{C-1}\sum_{i=0}^{R-1} S_{i,j}\,S_{i+v,j}^* = \theta_{\textbf{S}}(0,v) = 0.$$
Thus, the 2-dimensional array $\textbf{S}$ has perfect periodic autocorrelation.
\end{proof}

While Theorem \ref{thm:aop_perfect_seq} and Corollary \ref{thm:aop_array_perfect} establish the foundational equivalence between AOP sequences and perfect arrays, we need a method to bound the dimensions of these arrays. Using a result from B\"{o}mer and Antweiler \cite{bomer1992}, we can project the two-dimensional array down to one dimension. By taking the sum across the rows or columns, we compress the structural properties of the array into a single sequence over $\mathbb{C}$. Somewhat surprisingly, this projected one-dimensional sequence inherits the ``perfection'' of its parent array. We will use this invariant property in Section 3 to constrain the row and column dimensions independently, as bounding the length of this projected sequence strictly bounds the corresponding dimension of the array.

\begin{theorem}[B\"{o}mer and Antweiler \cite{bomer1992}]\label{thm:exist_cond_row_col}
Let $\textbf{S}_{\text{rows}} = \sum_{i=0}^{R-1}S_{i,j}$ be the sum of the rows of the $R \times C$ array $\textbf{S}$, and $\textbf{S}_{\text{cols}} = \sum_{j=0}^{C-1}S_{i,j}$ be the sum of the columns of $\textbf{S}$.
If $\textbf{S}$ is a perfect array, then $\textbf{S}_{\text{rows}}$ and $\textbf{S}_{\text{cols}}$ are perfect sequences.
\end{theorem}

\begin{proof}
Let $\textbf{r} = \textbf{S}_{\text{cols}}$ be the sequence formed by the sum of the columns of $\textbf{S}$, so that $r_i = \sum_{j=0}^{C-1} S_{i,j}$.
The periodic autocorrelation of $\textbf{r}$ at shift $\tau$ is given by
\begin{align*}
    \theta_{\textbf{r}}(\tau) &= \sum_{i=0}^{R-1} r_i r_{i+\tau}^* \\
    &= \sum_{i=0}^{R-1} \left( \sum_{j=0}^{C-1} S_{i,j} \right) \left( \sum_{k=0}^{C-1} S_{i+\tau, k}^* \right).
\end{align*}
Let $k = (j+h) \bmod C$. We can rewrite the double summation over $j$ and $k$ as a summation over $j$ and the column shift $h$
\begin{align*}
    \theta_{\textbf{r}}(\tau) &= \sum_{h=0}^{C-1} \sum_{i=0}^{R-1} \sum_{j=0}^{C-1} S_{i,j} S_{i+\tau, (j+h)\bmod C}^* \\
    &= \sum_{h=0}^{C-1} \theta_{\textbf{S}}(\tau, h),
\end{align*}
where $\theta_{\textbf{S}}(\tau, h)$ is the 2-dimensional periodic autocorrelation of the array $\textbf{S}$ at shift $(\tau, h)$. Because $\textbf{S}$ is a perfect array, $\theta_{\textbf{S}}(\tau, h) = 0$ for all $(\tau, h) \not\equiv (0,0) \mod{(R,C)}$. For any off-peak vertical shift $\tau \not\equiv 0 \mod R$, every term in the sum over $h$ is strictly off-peak in the 2-dimensional sense, meaning $\theta_{\textbf{S}}(\tau, h) = 0$ for all $h$. Thus, $\theta_{\textbf{r}}(\tau) = 0$ for all $\tau \not\equiv 0 \mod R$, proving that $\textbf{r}$ is a perfect sequence.\\

Furthermore, for the peak shift $\tau = 0$, we have $\theta_{\textbf{r}}(0) = \sum_{h=0}^{C-1} \theta_{\textbf{S}}(0, h)$. Since the array is perfect, $\theta_{\textbf{S}}(0, h) = 0$ for all $h \neq 0$. Thus, $\theta_{\textbf{r}}(0) = \theta_{\textbf{S}}(0, 0)$. \\

The proof that $\textbf{S}_{\text{rows}}$ is a perfect sequence with peak autocorrelation equal to the array's peak energy follows symmetrically.
\end{proof}

%
%

\section{AOP Bounds for Bivariate Polynomial Index Functions}

The core contribution of this paper is a proof that no sequence generated by a bivariate polynomial index function can satisfy the AOP for lengths exceeding the Frank-Heimiller bound of $n^2$. The fundamental mechanism preventing the existence of these extended sequences is algebraic periodicity. Because the sequence's phase transitions are dictated by an integer-coefficient polynomial evaluated modulo $n$, the phase trajectory is locked into a repeating geometric cycle. \\

Before establishing the main theorem (Theorem \ref{thm:aop_bound}), we must formalise this behavior. While Lemma \ref{lem:poly_mod} is a straightforward consequence of the binomial theorem, it acts as the bottleneck for all polynomial AOP constructions. It dictates that any attempt to expand the dimensions of an AOP array beyond $n$ columns will inevitably force columns to duplicate, fatally violating the mutual orthogonality requirement. \\

\begin{lemma}\label{lem:poly_mod}
Let $p(x)$ be a polynomial of any degree with integer coefficients. For any integer $i$ and any integer $n > 0$, we have
$$p(i+n) \equiv p(i) \mod n$$
\end{lemma}

\begin{proof}
Let $p(x) = \sum_{k=0}^m c_k x^k$, where $c_k \in \mathbb{Z}$. Evaluating the polynomial at $i+n$ gives
$$p(i+n) = \sum_{k=0}^m c_k (i+n)^k$$
By the binomial theorem, $(i+n)^k = i^k + k i^{k-1}n + \binom{k}{2} i^{k-2}n^2 + \dots + n^k$. Taking this expression modulo $n$, all terms containing a factor of $n$ vanish, leaving
$$(i+n)^k \equiv i^k \mod n$$
Substituting this back into the polynomial expansion yields
$$p(i+n) \equiv \sum_{k=0}^m c_k i^k \equiv p(i) \mod n$$
\end{proof}

Lemma \ref{lem:poly_mod} establishes that the phase values generated by any integer-coefficient polynomial index function must strictly repeat every $n$ elements when evaluated modulo $n$. We now apply this structural constraint to the AOP. The following theorem demonstrates that any perfect $n$-phase sequence constructed from a bivariate polynomial index function via the AOP cannot exceed the Frank-Heimiller length bound of $n^2$.

\begin{theorem}\label{thm:aop_bound}
Let $\textbf{s}$ be a sequence of length $L = R \times C$ over $n$ roots of unity, formed by enumerating an $R \times C$ array $\textbf{S}$ row-by-row.
Suppose the entries of $\textbf{S}$ are given by an index function $S_{i,j} = \omega^{p(i,j)}$, where $\omega = e^{2\pi \sqrt{-1}/n}$ and $p(x,y) \in \mathbb{Z}[x,y]$.
If $\textbf{s}$ possesses the Array Orthogonality Property (AOP) for the divisor $C$, then $L \leq n^2$.
\end{theorem}

\begin{proof}
The proof proceeds by establishing upper bounds on the number of columns $C$ and the number of rows $R$ of the array $\textbf{S}$.\\

First, we show that $C \leq n$. By the definition of the index function, the array entries are $S_{i,j} = \omega^{p(i,j)}$. Because $p(x,y)$ is a polynomial with integer coefficients, by Lemma \ref{lem:poly_mod} evaluating it at $j + n$ yields
$$p(i, j+n) \equiv p(i, j) \mod n.$$

Since $\omega$ is a primitive $n^{\text{th}}$ root of unity, the value of $\omega^k$ depends strictly on $k \mod n$. Therefore
$$S_{i, j+n} = \omega^{p(i, j+n)} = \omega^{p(i,j)} = S_{i,j}.$$

This implies that the sequence of column vectors is periodic with period $n$; specifically, the column $\textbf{S}[n]$ is identical to the column $\textbf{S}$. Condition 1 of the AOP mandates that any two distinct columns of $\textbf{S}$ must be orthogonal. If we assume $C > n$, then both $\textbf{S}$ and $\textbf{S}[n]$ exist in the array as distinct columns, requiring their cross-correlation $\theta_{\textbf{S}, \textbf{S}[n]}(\tau) = 0$ for all $\tau$. In particular, for $\tau = 0$, we require their inner product to be exactly zero. However, since $\textbf{S} = \textbf{S}[n]$, this inner product is simply the squared norm of $\textbf{S}$, which is strictly positive. A non-zero vector cannot be orthogonal to itself. Thus, by contradiction, we must have $C \leq n$. \\

Next, we show that $R \leq n$. Because the sequence $\textbf{s}$ possesses the AOP, by Corollary \ref{thm:aop_array_perfect} the associated $R \times C$ array $\textbf{S}$ is a perfect array. By Theorem \ref{thm:exist_cond_row_col}, the sequence formed by the sum of its columns, $\textbf{r} = [r_0, r_1, \dots, r_{R-1}]$, forms a perfect sequence of length $R$ where
$$r_i = \sum_{j=0}^{C-1} S_{i,j} = \sum_{j=0}^{C-1} \omega^{p(i,j)}.$$
Applying the same algebraic periodicity to the row index, by Lemma \ref{lem:poly_mod} we have $p(i+n, j) \equiv p(i,j) \mod n$. Consequently, $S_{i+n, j} = S_{i,j}$ for all $i,j$. It directly follows that $r_{i+n} = r_i$. \\

Assume, for the sake of contradiction, that $R > n$. Because $\textbf{r}$ is a perfect sequence of length $R$, its periodic autocorrelation is zero for all off-peak shifts $\tau \not\equiv 0 \mod R$. Since $R > n$, the shift $\tau = n$ is strictly off-peak, meaning $\theta_{\textbf{r}}(n) = 0$. However, computing this autocorrelation directly using the established periodicity $r_{i+n} = r_i$ yields
$$\theta_{\textbf{r}}(n) = \sum_{i=0}^{R-1} r_i r_{i+n}^* = \sum_{i=0}^{R-1} r_i r_i^* = \sum_{i=0}^{R-1} \|r_i\|^2.$$

For this sum of squared magnitudes to equal zero, we must have $r_i = 0$ for all $i$. If every element of $\textbf{r}$ is zero, its peak autocorrelation $\theta_{\textbf{r}}(0)$ must inherently also be zero. However, the proof of Theorem \ref{thm:exist_cond_row_col} establishes that the peak autocorrelation of the column-sum sequence is equal to the peak autocorrelation of the array itself, $\theta_{\textbf{r}}(0) = \theta_{\textbf{S}}(0,0)$. Since the elements of $\textbf{S}$ are complex roots of unity, their squared magnitudes are exactly $1$, meaning $\theta_{\textbf{S}}(0,0) = RC = L > 0$. This gives us the required contradiction: $\theta_{\textbf{r}}(0)$ cannot be simultaneously zero and strictly positive. Therefore, the assumption $R > n$ must be false, establishing that $R \leq n$. \\

Having established that $C \leq n$ and $R \leq n$, the length of the sequence $\textbf{s}$ is bounded by
$$L = R \times C \leq n^2.$$
\end{proof}

It is important to address a topological shift introduced by the column-sum projection in the proof of Theorem \ref{thm:aop_bound}. While the parent array $\textbf{S}$ is constructed over an $n$-phase alphabet---where every element is a complex root of unity with a constant magnitude of $1$---the projected sequence $\textbf{r}$ is formed by linearly summing these elements. The resulting elements $r_i$ no longer reside in the finite cyclic group of roots of unity; instead, they are cyclotomic integers that belong to the subring $\mathbb{Z}[\omega_n]$ within the continuous complex plane $\mathbb{C}$. Consequently, the elements of the projected sequence $\textbf{r}$ no longer possess a constant, unimodular magnitude.\\

In traditional digital signal processing, perfect sequences are often assumed to have a constant envelope due to physical hardware constraints. However, in the rigorous context of combinatorial sequence design, perfection is defined strictly by the topology of the periodic autocorrelation function. The contradiction established in Theorem \ref{thm:aop_bound} is entirely valid over $\mathbb{C}$ because it relies on the positive definiteness of the $L_2$ norm and the conservation of energy. Because the peak autocorrelation $\theta_{\textbf{r}}(0)$ exactly equals the total baseline energy of the array $\textbf{S}$, forcing the off-peak correlation to evaluate the sequence's own $L_2$ norm creates a contradiction and circumvents any need for the projected sequence to remain unimodular.

%
%

\section{AOP Bounds for Floored Rational Index Functions}

While simple polynomial index functions are locked to the $n^2$ limit by modular periodicity, modern sequence designers frequently utilise piecewise or step-like phase functions through the use of the floor function \cite{liu2004}\cite{blake2014}\cite{blake2016}. We now generalise our previously derived upper bound to piecewise index functions of this expanded geometric form.

\begin{theorem}\label{thm:aop_floor_bound}
Let $\textbf{s}$ be a sequence of length $L = R \times C$ over $K$ roots of unity, formed by enumerating an $R \times C$ array $\textbf{S}$ row-by-row. Suppose the entries of $\textbf{S}$ are given by an index function $S_{i,j} = \omega^{\lfloor p(i,j)/n \rfloor}$, where $\omega = e^{2\pi \sqrt{-1}/K}$, $p(x,y) \in \mathbb{Z}[x,y]$, and $n \in \mathbb{Z}^+$. If $\textbf{s}$ possesses the Array Orthogonality Property (AOP) for the divisor $C$, then $L \leq n^2 K^2$.
\end{theorem}
\begin{proof}
The proof proceeds by establishing upper bounds on the number of columns $C$ and the number of rows $R$ of the array $\textbf{S}$ by finding the expanded algebraic periodicity of the array.\\

First, we show that $C \leq nK$. By the definition of the index function, the array entries are $S_{i,j} = \omega^{\lfloor p(i,j)/n \rfloor}$. Because $p(x,y)$ is a polynomial with integer coefficients, expanding $p(i, j + nK)$ yields
$$p(i, j+nK) \equiv p(i, j) \mod{nK}.$$
Therefore, there exists some integer $M$ such that $p(i, j+nK) = p(i,j) + M(nK)$. Substituting this into the floor function gives
$$\left\lfloor \frac{p(i, j+nK)}{n} \right\rfloor = \left\lfloor \frac{p(i, j) + M(nK)}{n} \right\rfloor = \left\lfloor \frac{p(i,j)}{n} \right\rfloor + MK.$$
Since $\omega$ is a primitive $K^{\text{th}}$ root of unity, $\omega^{MK} = (\omega^K)^M = 1$. Therefore, the array entries satisfy
$$S_{i, j+nK} = \omega^{\lfloor p(i,j)/n \rfloor + MK} = \omega^{\lfloor p(i,j)/n \rfloor} = S_{i,j}.$$
This implies that the sequence of column vectors is periodic with an expanded period $nK$; specifically, the column $\textbf{S}[nK]$ is identical to the column $\textbf{S}$. Condition 1 of the AOP mandates that any two distinct columns of $\textbf{S}$ must be orthogonal. If we assume $C > nK$, then both $\textbf{S}$ and $\textbf{S}[nK]$ exist in the array as distinct columns, requiring their cross-correlation $\theta_{\textbf{S}, \textbf{S}[nK]}(\tau) = 0$ for all $\tau$. In particular, for $\tau = 0$, their inner product is exactly the squared norm of $\textbf{S}$, which is strictly positive. By contradiction, we must have $C \leq nK$.\\

Next, we show that $R \leq nK$. Because the sequence $\textbf{s}$ possesses the AOP, by Corollary \ref{thm:aop_array_perfect}  the associated $R \times C$ array $\textbf{S}$ is a perfect array. By Theorem \ref{thm:exist_cond_row_col}, the sequence formed by the sum of its columns, denoted $\textbf{r} = [r_0, r_1, \dots, r_{R-1}]$, is a perfect sequence of length $R$ where
$$r_i = \sum_{j=0}^{C-1} S_{i,j} = \sum_{j=0}^{C-1} \omega^{\lfloor p(i,j)/n \rfloor}.$$
Applying the same algebraic periodicity to the row index, we have $p(i+nK, j) = p(i,j) + M'(nK)$ for some integer $M'$. It directly follows that $S_{i+nK, j} = S_{i,j}$ for all $i,j$, and consequently, $r_{i+nK} = r_i$. \\

Assume, for the sake of contradiction, that $R > nK$. Because $\textbf{r}$ is a perfect sequence of length $R$, its periodic autocorrelation is zero for all off-peak shifts $\tau \not\equiv 0 \mod R$. Since $R > nK$, the shift $\tau = nK$ is strictly off-peak, meaning $\theta_{\textbf{r}}(nK) = 0$. However, computing this autocorrelation directly using the periodicity $r_{i+nK} = r_i$ yields
$$\theta_{\textbf{r}}(nK) = \sum_{i=0}^{R-1} r_i r_{i+nK}^* = \sum_{i=0}^{R-1} r_i r_i^* = \sum_{i=0}^{R-1} \|r_i\|^2.$$
For this sum of squared magnitudes to equal zero, every element of $\textbf{r}$ must be zero, which would mean its peak autocorrelation $\theta_{\textbf{r}}(0)$ is also zero. However, the proof of Theorem \ref{thm:exist_cond_row_col} establishes that the peak autocorrelation of the column-sum sequence equals the peak autocorrelation of the array itself, $\theta_{\textbf{r}}(0) = \theta_{\textbf{S}}(0,0) = RC = L > 0$. This poses a direct contradiction. Therefore, the assumption $R > nK$ must be false, establishing that $R \leq nK$. \\

Having established that $C \leq nK$ and $R \leq nK$, then the length of the sequence $\textbf{s}$ is bounded by
$$L = R \times C \leq n^2 K^2.$$
\end{proof}

%
%

\section{The Role of Period Collapse in Floored Index Functions}

Theorem \ref{thm:aop_floor_bound} reveals an interesting theoretical loophole: introducing a divisor $n$ inside the floor function algebraically expands the geometric bounding box of the array's periodicity from $K \times K$ to $nK \times nK$. This expanded bound suggests that it may be possible to construct a $K$-phase sequence of length $n^2 K^2$, directly challenging Mow's conjecture. However, while the expanded array size does not violate the cyclic wrapping conditions of the AOP, actually populating this space with orthogonal vectors requires neutralising highly volatile non-linear phase distortions.\\

To formalize this landscape, consider the structure of a general bi-quadratic index function inside the floor operation, $S_{i,j} = \omega^{\lfloor p(i,j)/n \rfloor}$, where $p(i,j) = A(j)i^2 + B(j)i + C(j)$ and $\omega = e^{2\pi \sqrt{-1}/K}$. We can evaluate the cross-correlation between any two distinct columns $j_1 \neq j_2$ by applying the mathematical identity $\lfloor x \rfloor = x - \{x\}$, where $\{x\}$ represents the discrete fractional part. This identity allows us to algebraically separate the cross-correlation summation into two distinct modulations
\begin{equation}\label{eq:biquad_sum}
\sum_{i=0}^{R-1} \exp\left( \frac{2\pi \sqrt{-1}}{nK} (\Delta A i^2 + \Delta B i + \Delta C) \right) \times \exp\left( -\frac{2\pi \sqrt{-1}}{K} \left( \left\{ \frac{p(i,j_1)}{n} \right\} - \left\{ \frac{p(i,j_2)}{n} \right\} \right) \right) = 0
\end{equation}
where $\Delta A = A(j_1) - A(j_2)$, and similarly for $\Delta B$ and $\Delta C$. 

The first exponential term in Equation \ref{eq:biquad_sum} generates a continuous Gaussian phase evolution over the artificially expanded period $nK$. In isolation, Gaussian sums are highly structured and are well-known to yield orthogonal vectors. However, the second exponential isolates the non-linear fractional components. \\

To understand why this fractional part is so volatile, we must consider the instantaneous frequency of the argument $p(i,j)/n$. Because $p(i,j)$ contains a quadratic term $A(j)i^2$, the ``frequency'' of the fractional wrapping grows linearly with the row index $i$. Numerically, this behaves like a discrete chirp signal sampled beyond its Nyquist rate. For small values of $i$, the fractional part $\{p(i,j)/n\}$ forms a slow, predictable staircase. But as $i$ increases, the distance between discrete phase jumps shrinks rapidly. Analogous to Weyl's criterion for equidistribution, the quadratic fractional sequence pseudo-randomly scatters across the available discrete values $0, 1/n, 2/n, \dots, (n-1)/n$. \\

When this high-frequency, pseudo-random noise term is multiplied against the underlying continuous Gaussian sum, it breaks the geometric symmetry required for the sum to perfectly cancel to zero. Instead of the phase vectors gracefully closing into a perfect polygon (which mathematically guarantees orthogonality), the discrete fractional jumps cause the cross-correlation to undergo a random walk in the complex plane. \\

For the cross-correlation to remain exactly zero, these highly irregular fractional jumps must not destructively interfere with the orthogonality of the underlying Gaussian sum. To systematically prevent this fractional phase distortion from destroying the column orthogonality across all $C$ columns, the quadratic contribution to the fractional part must algebraically cancel for all row indices $i$. \\

Cancelling this quadratic scattering strictly forces the constraint that the leading coefficient must be a multiple of the divisor, i.e., $A(j) \equiv 0 \pmod n$ for all possible column indices $j$. If we let $A(j) = n \cdot A'(j)$ for some integer polynomial $A'(j)$, we can substitute this constrained coefficient back into the continuous Gaussian phase term

$$\exp\left( \frac{2\pi \sqrt{-1}}{nK} (\Delta A i^2) \right) = \exp\left( \frac{2\pi \sqrt{-1}}{nK} (n \Delta A' i^2) \right) = \exp\left( \frac{2\pi \sqrt{-1}}{K} \Delta A' i^2 \right).$$

Notice that the divisor $n$ completely cancels out of the quadratic phase equation. The period of the Gaussian sum instantly collapses from its artificially expanded state $nK$ back to the original base phase $K$. Because the number of mutually orthogonal quadratic sequences of algebraic period $K$ cannot exceed $K$, the number of columns $C$ is fundamentally restricted to $K$. \\

Thus, we arrive at an inescapable structural paradox where the very coefficient constraints required to silence the destructive fractional scattering and maintain orthogonality inherently force the geometric bounds of the array to revert to $K \times K$. Thus, the expanded orthogonal capacity of $n^2K^2$ almost certainly collapses to the Frank-Heimiller limit of $K^2$.

%
%

\section{The Commutative Boundary and Non-Commutative Exceptions}

The bounds established in Theorems \ref{thm:aop_bound} and \ref{thm:aop_floor_bound} limits sequence discovery under the AOP. However, the validity of these bounds relies on the mathematical assumption of commutativity within the underlying phase space. It is important to note that the periodicity bottleneck and the subsequent collapse of floored expansions rely intrinsically on the abelian group structure of complex roots of unity. Therefore, these bounds do not govern sequences defined over non-commutative algebras. \\

In the non-commutative group of simple unit quaternions $Q_8 = \{\pm 1, \pm \mathbf{i}, \pm \mathbf{j}, \pm \mathbf{k}\}$, the algebraic assumption that fractional phase components can be cleanly separated and collapsed modulo $n$ (as utilised in our analysis of fractional phase scattering) fails. The chaotic fractional scattering that destroys cross-correlation orthogonality in the complex plane can be mathematically absorbed or perfectly offset by the anti-commutative properties of the elements themselves (where $\mathbf{i}\mathbf{j} = -\mathbf{j}\mathbf{i}$). This non-abelian cancellation allows orthogonal vectors to exist well beyond the $K^2$ dimension limit. \\

Indeed, while early investigations by the author hypothesised a quaternionic extension of Mow's conjecture \cite{blake2016}, recent breakthroughs by Bright et al. \cite{bright2019} have disproven this conjecture and constructed interesting new sequences over the unit quaternions. Using combinatorial design theory and Williamson matrices, Bright et al. proved that perfect and odd-perfect quaternion sequences exist in unbounded lengths $2^t$ for $t \ge 0$, actively bypassing the $K^2$ boundary constraints that dominate commutative phase spaces.

%
%

\section{Conclusion}

A universal, unconditional proof of Mow's conjecture---that no perfect $n$-phase sequence of length strictly greater than $n^2$ can exist---remains an active pursuit. However, the theoretical bounds established in this paper significantly tighten the structural constraints for sequences constructed via the AOP. \\

We have demonstrated that traditional avenues for constructing long perfect sequences utilising integer polynomials are inherently locked to their algebraic period, confirming the $n^2$ bound. Furthermore, we formalised the structural limitations of floored rational index functions, proving that any attempt to geometrically stretch the array dimensions triggers chaotic fractional phase scattering. The specific coefficient constraints required to silence this destructive scattering inevitably force the continuous phase space to collapse back to the base field, re-establishing the Frank-Heimiller limit. \\ 

Critically, we emphasise that these rigid structural bottlenecks rely entirely on the abelian nature of complex roots of unity. When these commutative constraints are removed, as seen in recent quaternionic sequence discoveries, the bounds can be vastly exceeded. Future work seeking to universally resolve Mow's conjecture in the complex plane must bridge the gap between the algebraic geometry bounds proven here and the complementary character theory limits of difference sets to definitively rule out the existence of completely aperiodic, non-algebraic sequence designs.

\end{document}